\documentclass[runningheads]{llncs}
\usepackage{amsmath}
\usepackage{amsfonts}
\usepackage{graphicx}
\usepackage{tikz}
\usepackage{parskip}
\usepackage{verbatim}
\usetikzlibrary{calc,shapes,patterns}

\newtheorem{observation}[theorem]{Observation}
\newtheorem{claim2}[theorem]{Claim}

\begin{document}

\title{Planar projections of graphs}
\author{N.R. Aravind\inst{1}\orcidID{0000-0002-6590-7952} \and Udit Maniyar\inst{1}}
\institute{$^1$Indian Institute of Technology Hyderabad}
\date{}

\maketitle

\begin{abstract}
We introduce and study a new graph representation where vertices are embedded in three or more dimensions, and in which the edges are drawn on the projections onto the axis-parallel planes.
We show that the complete graph on $n$ vertices has a representation in $\lceil \sqrt{n/2}+1 \rceil$ planes. In 3 dimensions, we show that there exist graphs with $6n-15$ edges that can be projected onto two orthogonal planes, and that this is best possible. Finally,
we obtain bounds in terms of parameters such as geometric thickness and linear arboricity. Using such a bound, we show that every graph of maximum degree 5 has a plane-projectable representation in 3 dimensions.

\end{abstract}

\section{Introduction}

In this paper, we consider embeddings of graphs where the vertices are mapped to points in $\mathbb{R}^d$, for $d \geq 3$, and the edges are represented by line-segments on the $\dbinom{d}{2}$ axis-parallel planes. For example, a 3-dimensional network may be visualized by placing it inside a cube and drawing the edges on the walls of the cube by projecting the points.

One motivation is the connection to two classical parameters, thickness and geometric thickness.
The thickness of a graph $G$, is the smallest number of planar subgraphs into which the edges of $G$ can be decomposed. This was introduced by Tutte in \cite{T63}; see also \cite{MOS98} for a survey of thickness.
Geometric thickness adds the restriction that all the subgraphs must be embedded simultaneously, that is, with a common embedding of the vertices. This was studied in \cite{DEH00} for complete graphs.
The connection between geometric thickness and parameters such as maximum degree and tree-width has been studied in various papers: \cite{DEK04},\cite{BMW06},\cite{D11}.
While using the standard co-ordinate planes in high dimensions is more restrictive than thickness, it appears to be less so than geometric thickness (Section 3).

Book embeddings, defined by Ollmann in \cite{O73}, are restrictinos of geometric drawings in which the vertices are in convex position.
The book thickness of $G$ is the smallest number of subgraphs that cover all the edges of $G$ in such a drawing. This is also known as stack number, and is studied in \cite{DW05}.
Also see \cite{DW04} for a survey. More generally, a survey on simultaneous embedding of graphs may be found in \cite{BKR13}.

In \cite{HSV99}, the authors showed that $n$-vertex graphs of geometric thickness 2 can have at most $6n-18$ edges. Such graphs can also be represented as projections in two orthogonal planes; orthogonal planes appear to allow a greater degree of freedom, as we give a construction of graphs with $6n-15$ edges.
We also note that a plane-projectable construction with $6n-17$ edges was given in \cite{M16}.



\subsection{Preliminaries:}

For a point $q$ in $\mathbb{R}^d$, we denote by $\pi_{i,j}(q)$, the projection of $q$ on the plane $\{x\in\mathbb{R}^d\mid x_i=x_j=0\}$ formed by the $i$th  and $j$th co-ordinate axes.

\begin{definition}
Given a graph $G=(V,E)$, we say that an injective map $\pi: V \to \mathbb{R}^d$ is a {\bf plane-projecting map} of $G$ if there exists a decomposition 
$E=\cup_{1 \leq i<j \leq d} E_{i,j}$ such that the projection $\pi_{i,j}$ is injective and induces a straight-line planar embedding of the subgraph $(V,E_{i,j})$.
\end{definition}

We define the {\bf plane-projecting dimension} of a graph $G$ to be the smallest integer $d$ for which a plane-projecting map in $\mathbb{R}^d$ exists for $G$.
We denote this by $pdim(G)$.

If $pdim(G) \leq d$, we shall say that $G$ is {\bf $d$-dimensionally projectable} or {\bf plane-projectable in $d$-dimensions}.

We note the following connection between the plane-projecting dimension and the two thickness parameters of a graph.

\begin{observation}\label{pdim-vs-thickness}
Let $G$ have thickness $\theta(G)=r$ and geometric thickness $\bar{\theta(G)}=s \geq r$.
Then we have:

(i) $\dfrac{\sqrt{8r+1}+1}{2} \leq pdim(G) \leq 2r$.

(ii) $pdim(G) \leq 2 \lceil \sqrt{s} \rceil$.
\end{observation}

\begin{proof}
The first inequality in (i) follows from the observation that 
$r \leq \dbinom{pdim(G)}{2}$; the second inequality is easy to see.
For (ii), we let $k=\lceil \sqrt{s} \rceil$.
For a vertex $v$, let $(a,b)$ be its position in an optimal geometric thickness representation of $G$.
Then the map $f(v)=(a,a,\ldots,a,b,b, \ldots, b)$ (with number of $a$'s and $b$'s each equal to $k$), is a plane-projecting map, with the edge sets $\{E_{i,j}:1 \leq i \leq k,k+1 \leq j \leq 2k\}$ corresponding to the subgraphs of the geometric thickness representation, and $E_{i,j}$ drawn on the plane with $x_i=x_j=0$.
\qed
\end{proof}

In \cite{D11}, the author obtained a bound of $O(\log n)$ on the geometric thickness of graphs with arboricity two; thus as a corollary, we obtain a bound of $O(\sqrt{\log n})$ on the plane-projecting dimension of such graphs.

\subsection{Our Results:}

In Section 2, we obtain an upper bound of $\sqrt{n/2}+O(1)$ on the plane-projecting dimension of $K_n$.

In Section 3, we give a construction of graphs having $n$ vertices and $6n-15$ edges that can be projected on two orthogonal planes, and further show that this is tight. We also obtain an upper bound on the maximum number of edges of a $n$-vertex graph that is plane-projectable in 3 dimensions.

In Section 4, we show that every graph of maximum degree five is plane- projectable in three dimensions, by obtaining an upper bound in terms of the linear arboricity of $G$ (which is the minimum number of linear forests that partition the edges of $G$).
We also obtain a general upper bound of $\Delta(G)\left(\dfrac{1}{2}+o(1)\right)$ on $pdim(G)$.
Note that an upper bound of $\Delta(G)+1$ follows from Observation \ref{pdim-vs-thickness} and a result of \cite{H91}, which states that the thickness of a graph of maximum degree $\Delta$ is at most $\lceil \dfrac{\Delta}{2} \rceil$.

\section{Plane-projecting dimension of complete graphs}

In the paper \cite{DEH00}, the authors show that the geometric thickness of $K_n$ is at most $\lceil n/4 \rceil$. 
Combining this with Observation \ref{pdim-vs-thickness}, we get $pdim(K_n) \leq \lceil \sqrt{n} \rceil$.

The thickness of $K_n$ is known to be 1 for $1 \leq n \leq 4$, 2 for $5 \leq n \leq 8$, 3 for $9 \leq n \leq 10$,
and $\lceil \dfrac{n+2}{6} \rceil$ for $n>10$.
Thus, for $n>10$, we get $pdim(K_n) \geq \sqrt{n/3} $.

By using the construction of \cite{DEH00} in a more direct way, we obtain the following improved upper bound.
\begin{theorem}\label{pdim-complete}
$pdim(K_n) \leq \lceil \dfrac{\sqrt{2n+7}+1}{2} \rceil$.
\end{theorem}

To prove Theorem \ref{pdim-complete}, we shall use the following lemma, which we first state and prove.

\begin{lemma}\label{convex-projection}
For every natural number $d \geq 2$ and every natural number $n$, there exist $n$ points in $\mathbb{R}^d$ such that for every $1\leq i<j \leq d$, the projections of these points to the $(i,j)$-plane are in convex position, and in the same order on the convex hull.
\end{lemma}

\begin{proof}[of Lemma \ref{convex-projection}]
We consider the point-set 
$P_i=(\cos (a_i+ib),\cos (a_i+(i+1)b),\ldots,\cos(a_i+(i+d-1)b))$ for some suitable $b$ and $a_i$s. For given $i,j \in \{1,2,\ldots,d\}$, the projection of these points in the $(i,j)$ plane lie on an ellipse.
\qed
\end{proof}

We now prove Theorem \ref{pdim-complete}.

\begin{proof}[of Theorem \ref{pdim-complete}]
Let $d$ be such that $\dbinom{d}{2} \geq \lceil \dfrac{n}{4} \rceil$.
We assume that $n=2k$, where $k$ is even, and find sets $S,T$ of $n/2$ points each in $\mathbb{R}^d$ such that
the projections of $S$ and $T$ to two given axis-parallel planes lie on an ellipse, with the ellipse corresponding to $S$ always contained inside and congruent to the ellipse corresponding to $T$,
as shown in Figure 1 (right).

\begin{figure}
\centering
\begin{tikzpicture}[scale=0.6,rotate=45,auto=center,every node/.style={circle,inner sep=1.5pt}]
\draw[black](0,0) ellipse (4 and 2);
\node[fill=black] (v1) at (-4,0){};
\node[fill=black] (v5) at (4,0){};
\node[fill=black] (v3) at (0,-2){};
\node[fill=black] (v7) at (0,2){};
\node[fill=black] (v2) at (-2,-1.732){};
\node[fill=black] (v8) at (-2,1.732){};
\node[fill=black] (v4) at (2,-1.732){};
\node[fill=black] (v6) at (2,1.732){};

\draw (v1)--(v5)--(v2)--(v4)--(v3);
\draw (v1)--(v6)--(v8)--(v7);

\node[fill=none] at (-4.5,0){$v_1$};
\node[fill=none] at (4.5,0){$v_5$};
\node[fill=none] at (-2,-2.2){$v_2$};
\node[fill=none] at (0,-2.5){$v_3$};
\node[fill=none] at (2,-2.2){$v_4$};
\node[fill=none] at (2,2.2){$v_6$};
\node[fill=none] at (0,2.5){$v_7$};
\node[fill=none] at (-2,2.2){$v_8$};

\draw[black](5,-5) ellipse (4 and 2);
\node[fill=black] (w5) at (9,-5){};
\node[fill=none] at (10,-5){$w_5$};
\node[fill=black] (w1) at (1,-5){};
\node[fill=none] at (0.5,-5){$w_1$};


\draw[black](5,-4.5) ellipse (0.5 and 0.25);
\node[fill=none] (x1) at (5,-4.2){};
\node[fill=none] (x2) at (5,-4.75){};
\node[fill=none] (x3) at (4.6,-4.5){};
\node[fill=none] (x4) at (5.4,-4.5){};

\draw (w1)--(w5);
\draw (w1)--(x1)--(w5);
\draw (w1)--(x2)--(w5);
\draw (w1)--(x3);
\draw (w5)--(x4);

\end{tikzpicture}
\caption{Left: Path in $S$/$T$; Right: Edges between diametrically opposite vertices of $T$ and the vertices of $S$}
\end{figure}
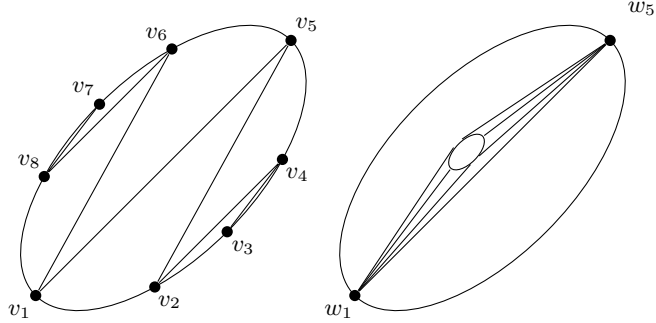

The drawing of edges is now the same as in \cite{DEH00}, which we explain for the sake of completeness.

We decompose the complete graph on each of $S,T$ into $k/2$ disjoint paths and draw their edges in $k/2$ planes such that each path contains exactly one pair of diametrically opposite vertices that are adjacent. Here, we use the phrase ``diametrically opposite" for a pair of vertices if the points corresponding to them have exactly $k/2-1$ other points between them on the convex hull.
This is illustrated in Figure 1 (a), where $v_1,v_5$ are the diametrically opposite pair which are adjacent. Other diametrically opposite pairs are $\{v_2,v_6\},\{v_3,v_7\}$ etc, each of which shall be adjacent in a different plane.

Finally, we add edges between every vertex of $S$ and the diametrically opposite pair of $T$. That this can be done is shown in \cite{DEH00}, by showing that there exist a set of parallel lines each passing through one point in $S$, and arguing by continuity that if the diametrically opposite pair is far enough, they may be joined to the points of $S$ without intersections. This is illustrated in Figure 1 (b).
\qed
\end{proof}

\section{Plane-projectable graphs in $\mathbb{R}^3$}

In this section, we focus on $\mathbb{R}^3$, and ask the following two extremal questions.

Q1. What is the maximum number of edges of a $n$-vertex graph with plane-projecting dimension three?

Q2. What is the maximum number of edges of a $n$-vertex graph whose edges can be projected into two orthogonal planes in $\mathbb{R}^3$?

We shall answer Question 1 partially by giving an upper bound of $9n-24$, and Question 2 completely, by giving matching upper and lower bounds.

As mentioned earlier, any graph with geometric thickness two can be projected in two of the co-ordinate planes.
In \cite{HSV99}, it was shown that a $n$-vertex graph of geometric thickness two can have at most $6n-18$ edges and that $6n-20$ edges
is achievable. This was improved in \cite{DGM16}, where it was shown that for every $n \geq 9$, $6n-19$ edges is achievable.

By modifying their construction, we show the following:

\begin{theorem}
For every $n \geq 14$, there exists a graph $G_n$ on $6n-15$ edges with an embedding in $\mathbb{R}^3$ such that the restriction to two of the planes form planar straight-line embeddings of two graphs whose edge-union is equal to $G_n$.
\end{theorem}

\begin{figure}[h]
\centering
\includegraphics[width=6cm]{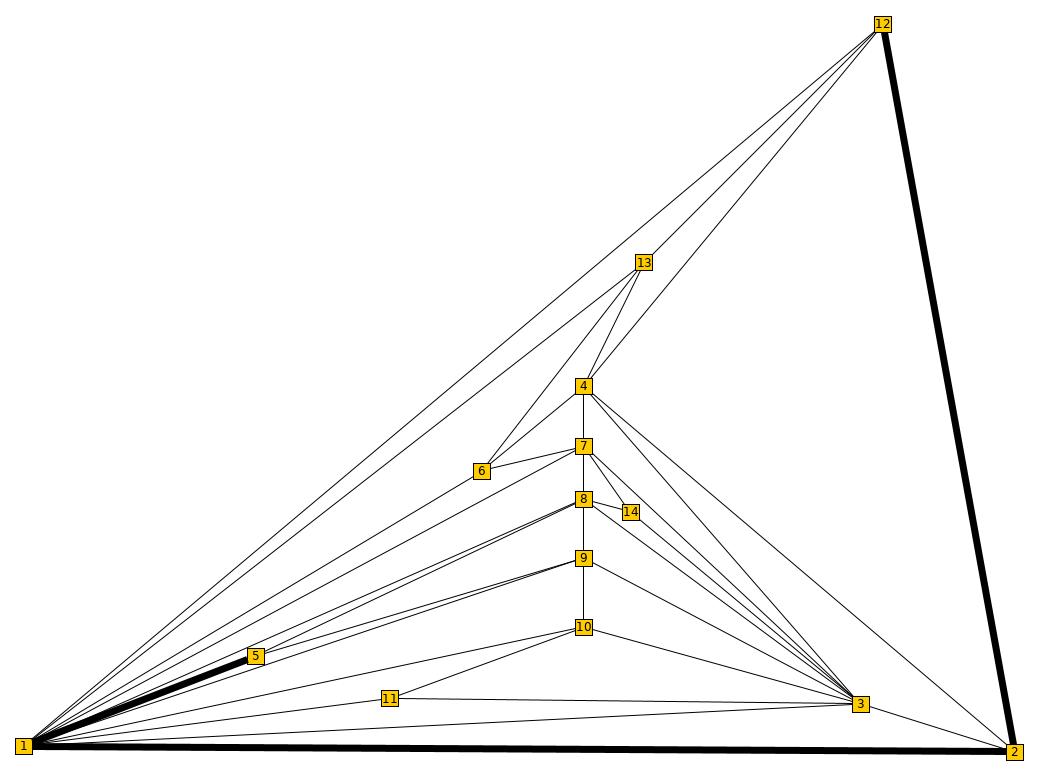}
\includegraphics[width=6cm]{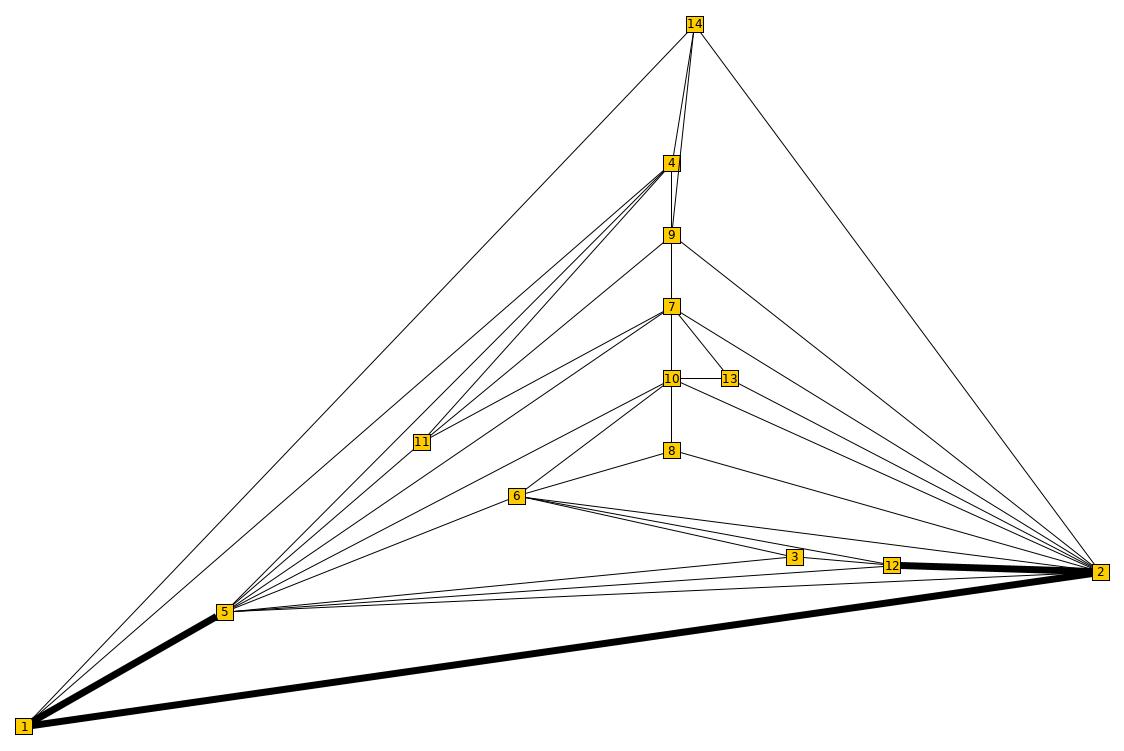}
\caption{On left: $H_{14}$, On right: $M_{14}$; $G_{14}$ has $6 \times 14-15 = 69$ edges. The dark edges are common to both planes. Also the exact vertex positions are not important, but the ordering of the vertices on the common axis should be the same.}

\label{fig:max-edge-fig}
\end{figure}

\begin{proof}

Let $H_n$, $M_n$ be the projection of $G_n$ on $XY$, $YZ$ planes respectively. Observe that if we fix the embedding of vertices of $G$ in $H$, then in $M$ we would have the freedom to move vertices along $Z$ axis because $z$ coordinate of vertices have not been fixed yet.

Figure \ref{fig:max-edge-fig} gives the construction of a graph $G_{14}$ with 14 vertices and $6 \times 14 - 15 = 69$ edges.

Let us assume that we are given $H_k, M_k$ which are the planar projections of $G_k$ on $XY, YZ$ planes respectively, such that $|E_k| = 6k-15$.

We now show that we can add a vertex $v_{k+1}$ with 6 neighbors, such that three of the new edges are added to $H_k$ to obtain $H_{k+1}$ and the other three are added to $M_k$ to obtain $M_{k+1}$.

In $H_k$, we place $v_{k+1}$ inside a triangle whose vertices are disjoint from the vertices present on the convex hull of $M_k$ namely $v_1, v_2, v_k$. Now the $(x,y)$ coordinates of $v_{k+1}$ are fixed we can take the $z$ coordinate of $v_{k+1}$ to be a value strictly greater than the $z$ coordinate of $v_k$. Now in $M_k$($YZ$ plane) we can add $v_{k+1}$ by connecting $v_{k+1}$ to $v_1,v_2, v_{k}$.

We can always find a triangle whose vertices should not contain any one of $v_1, v_2, v_k$. If $k$ is odd we add $v_{k+1}$ inside the triangle $v_6, v_4, v_{k-1}$ in $H_k$ and if $k$ is even we add $v_{k+1}$ inside the triangle $v_7,v_8, v_{k}$ in $H_k$.

Since we have added 6 edges to the graph $G_k$ the new graph $G_{k+1}$ contains $6k-15 + 6 = 6(k+1) - 15$. The vertex $v_{n+1}$ is also being mapped to a suitable point in $\mathbb{R}^3$.

Inductively using the above process, we can generate $G_n$ for all n such that $|E_n| = 6n-15$.

\end{proof}

We will now show that the above upper bound is in fact tight;
we first need the following definition.

\begin{definition}
We say that an embedding of a planar graph $G$ is maximally planar if no non-adjacent pair of vertices can be joined by a line-segment without intersecting the existing edges.
\end{definition}

\begin{theorem}\label{upper-bound-2planes}
Let $G$ be a connected graph on $n \geq 3$ vertices having an embedding in $\mathbb{R}^3$ such that the restriction to two of the planes form straight-line planar embeddings of two graphs. Then $G$ has at most $6n-15$ edges.
\end{theorem}

\begin{proof}
Consider an embedding of $G(V, E)$ such that the edges of $G$ are covered by planar drawings in two (projected) planes. Let $XY$ and $YZ$ be the two planes and let $G_1=(V,E_1)$ and $G_2=(V,E_2)$ be the two planar sub-graphs respectively, which are projected on these planes.

Clearly we can assume that the embeddings of both $G_1$ and $G_2$ are maximally planar.

Let \\
$A$ to be the vertex with lowest $y$ coordinate value,\\
$B$ to be the vertex with highest $y$ coordinate value,\\
$C$ to be the vertex with second lowest $y$ coordinate value,\\
$D$ to be the vertex with second highest $y$ coordinate value. 

\begin{claim2}\label{non-intersecting}
Both $AC$ and $BD$ belong to $G_1 \cap G_2$.
\end{claim2}

{\bf Proof of Claim \ref{non-intersecting}:}
Suppose for contradiction that $AC \notin G_1$. Since $G_1$ is maximally planar, there must be an edge $uv$ that intersects the line-segment joining $AC$. Therefore the $y$ co-ordinate of $u$ or the $y$ co-ordinate of $v$ must lie between the $y$ co-ordinates of $A$ and $C$, which contradicts the choice of $A,C$. The proof for $BD$ is identical.

We now consider two cases.

{\bf Case 1:} $|E_1|< 3n-6$ or $|E_2| < 3n-6$. 
In this case, we have: $|E_1 \cup E_2|=|E_1|+|E_2|-|E_1 \cap E_2| \leq 6n-13-2=6n-15$.

{\bf Case 2:} $|E_1|=|E_2|=3n-6$.
In this case, we show that in addition to $AC$ and $BD$, the edge $AB$ also belongs to both $E_1$ and $E_2$, which shows that $|E_1 \cup E_2| \leq 6n-15$.

Since $G_1$ has $3n-6$ edges, its convex hull is a triangle. By the definition of $A,B$, we see that $AB$ should be on the convex hull, and hence is an edge of $G_1$. Similarly, $AB$ belongs to $G_2$ as well.

This completes the proof of Theorem \ref{upper-bound-2planes}. \qed
\end{proof}

We now give an upper bound on graphs with plane-projectable dimension three.

\begin{theorem}\label{upper-bound-3planes}
Let $G$ be a connected graph on $n \geq 3$ vertices having an embedding in $\mathbb{R}^3$ such that the restriction to  the three planes form straight-line planar embeddings of three graphs. Then $G$ has at most $9n-24$ edges.
\end{theorem}
\begin{proof}
Consider an embedding of $G(V, E)$ such that the edges of $G$ are covered by planar drawings in three (projected) planes. Let $XY$, $YZ$ and $ZX$ be the two planes and let $G_1=(V,E_1)$, $G_2=(V,E_2)$ and $G_3=(V,E_3)$ be the three planar sub-graphs respectively, which are projected on these planes.

We may assume that $G_1, G_2, G_3$ are maximally planar.\\
Here we have to consider few cases:

\textbf{Case 1}: $|E_1| = |E_2| = |E_3| = 3n-6$. In this case we use the same argument as Theorem \ref{upper-bound-2planes}, we get $|E_2 \setminus E_1| = 3n-9$, $|E_3 \setminus E_1| = 3n-9$.

$|E_1 \cup E_2 \cup E_3| \leq |E_1| + |E_2 \setminus E_1|  + |E_3 \setminus E_1|$.

$\implies |E_1 \cup E_2 \cup E_3| \leq (3n-6) + (3n-9) + (3n-9) = 9n-24$.

\textbf{Case 2}: $|E_1| = |E_2| = 3n-6, |E_3| \leq 3n-7$. In this case if we use the same argument as Theorem \ref{upper-bound-2planes} we get $|E_2 \setminus E_1| = 3n-9$.

Since $G_1, G_3$ are maximally planar using Claim \ref{non-intersecting}, we get $|E_3 \setminus E_1| \leq 3n-7 -2 = 3n-9$.

$\implies |E_1 \cup E_2 \cup E_3| \leq (3n-6) + (3n-9) + (3n-9) = 9n-24$.

\textbf{Case 3}:$|E_1| =3n-6, |E_2| \leq 3n-7, |E_3| \leq 3n-7$.

Since $G_1, G_2, G_3$ are maximally planar using Claim \ref{non-intersecting}, we get $|E_3 \setminus E_1|\leq 3n-7-2=3n-9$.

$|E_2 \setminus E_1| \leq 3n-7 -2 = 3n-9$.

$\implies |E_1 \cup E_2 \cup E_3| \leq (3n-6) + (3n-9) + (3n-9) = 9n-24$.

\textbf{Case 4}:$|E_1| \leq3n-7,  |E_2| \leq 3n-7, |E_3| \leq 3n-7$.

Since $G_1, G_2, G_3$ are maximally planar, using Claim \ref{non-intersecting}, we get $|E_3 \setminus E_1| \leq
|E_2 \setminus E_1| \leq 3n-7 -2 = 3n-9$.
Similarly $|E_3 \setminus E_1 | = \leq 3n-9$;

$\implies |E_1 \cup E_2 \cup E_3| \leq (3n-7) + (3n-9) + (3n-9) = 9n-25$.

This completes the proof of Theorem \ref{upper-bound-3planes}.
\end{proof}

\section{Relation with linear arboricity and maximum degree}
A linear forest is a forest in which every tree is a path. 
The linear arboricity of a graph $G$ is the minimum number of linear forests into which the edges of $G$ can be decomposed.

We have the following.


\begin{proposition}\label{la}
If $G$ has linear arboricity at most $k$, then there is an embedding of $G$ in $\mathbb{R}^{k}$ such that the edges of $G$ can be drawn on the following $k$ standard planes: for $i=1,\ldots,k-1$, the $i$th plane is $\{x \in \mathbb{R}^{k}: x_j=0 \forall j \notin \{1,i\}\}$, and the $k$th plane is $\{x \in \mathbb{R}^k: x_j=0 \forall j \notin \{2,3\}\}$. In particular, $pdim(G) \leq k$.
\end{proposition}

In \cite{EP84}, it was shown that graphs of maximum degree 5 have linear arboricity at most 3.
Thus, we get the following.

\begin{corollary}
Any graph of maximum degree 5 is plane-projectable.
\end{corollary}

In \cite{A88}, Alon showed that a graph of maximum degree $\Delta$ has linear arboricity at most $\dfrac{\Delta}{2}+o(\Delta)$.
Thus, we have: 
$pdim(G) \leq \Delta(G)\left(\dfrac{1}{2}+o(1)\right)$.

We shall actually prove a stronger form of Proposition \ref{la}, in which we replace linear arboricity by caterpillar arboricity, which we define below.

A caterpillar tree is a tree in which all the vertices are within distance 1 of a central path.
A caterpillar forest is a forest in which every tree is a caterpillar tree.
The {\it caterpillar arboricity} of a graph $G$ is the minimum number of caterpillar forests into which the edges of $G$ can be decomposed. This has been studied previously in \cite{GO09}.

The main idea behind Proposition \ref{la} is the following.
\begin{lemma}\label{hc-order}
Given a caterpillar tree $G$ with vertex set $V=\{v_1,v_2,\ldots,v_n\}$, and $n\geq 2$ distinct real numbers $y_1,\ldots,y_n$, there exist $n$ real numbers $x_1,\ldots,x_n$
such that $G$ has a straight-line embedding with the vertex $v_i$ mapped to $(x_i,y_i)$.
\end{lemma}
\begin{proof}

Let $w_1, w_2, \ldots w_k$ be the vertices on the central path such that $w_i$ has an edge with $w_{i-1}$ and  $w_{i+1}$. All the indices are taken modulo $k$.
Also, let $L_i$ denote the set of leaf vertices adjacent to $w_i$.

We set the $x$ co-ordinate of $w_i$ to be $i$, and the $x$ co-ordinate of every vertex of $L_i$ to be $i+1$.
Clearly the edges of the caterpillar drawn with the above embedding are non-crossing. \qed
\end{proof}

We remark that the above result and concept have also been studied as "unleveled planar graphs" in \cite{EFK06}.

\begin{lemma}
Given a cycle $G$ with vertex set $V=\{v_1,v_2,\ldots,v_n\}$, and $n\geq 2$ distinct real numbers $y_1,\ldots,y_n$, there exist $n$ real numbers $x_1,\ldots,x_n$
such that $G$ has a straight-line embedding with the vertex $v_i$ mapped to $(x_i,y_i)$.
\end{lemma}
\begin{proof}
Without loss of generality, let $v_1, v_2, \ldots v_n$ be the vertices on the cycle such that $v_i$ has an edge with $v_{i-1}$ and $v_{i+1}$ and $v_1$ be the vertex with smallest y coordinate value. All the indices are taken modulo $n$.

We first remove the edge between $v_1$ and $v_n$ so that the remaining graph is a path for which we use the previous lemma to construct the embedding.

Now to add back the edge $v_1v_n$, we have to make sure that none of the other edges intersect with the edge between $v_1$ and $v_n$. 
Let $slope_i$ to be the slope between $v_1$ and $v_i$, and note that this is positive for all $i$, since $v_1$ has the lowest y coordinate. Let $m=min_{i}\{slope_i\}$. We now draw a line $L$ through $v_1$ with slope less than $m$ and place the vertex $v_n$ on $L$, as illustrated in the figure below.


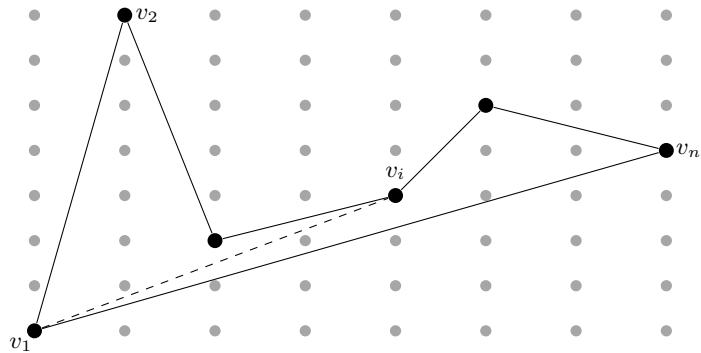
\begin{figure}
\centering
\begin{tikzpicture}[scale=.6,auto=center,every node/.style={circle,inner sep=1.5pt}]
 \foreach \x in {0,2,...,14}
 \foreach \y in {0,1,...,7}
      \node[fill=black!35] at (\x,\y){} ;

  \node [fill=black,inner sep=2pt] (v1) at (0,0){};
 \node [fill=black,inner sep=2pt] (v2) at (2,7){};
  \node [fill=black,inner sep=2pt] (v3) at (4,2){};
  \node [fill=black,inner sep=2pt] (vi) at (8,3){};
 \node [fill=black,inner sep=2pt] (vj) at (10,5){};
 \node [fill=black,inner sep=2pt] (vn) at (14,4){};
 
 \draw (v1)--(v2)--(v3);
 \draw (v3)--(vi)--(vj);
 \draw (vj)--(vn)--(v1);
 \draw[dashed] (v1)--(vi);
 
 \node [fill=none] at (-0.3,-0.3) {$v_1$};
 \node [fill=none] at (2.5,7) {$v_2$};
 \node [fill=none] at (8,3.5) {$v_i$};
 \node [fill=none] at (14.5,4) {$v_n$};
 
\end{tikzpicture}
\caption{Cycle graph with given $y$ co-ordinates}
\end{figure}
\end{proof}

The proposition below also follows from an application of Lemma \ref{hc-order}.

\begin{proposition}
Let $G$ be the edge-union of a planar graph and $d$ paths. Then $pdim(G) \leq d+2$.
\end{proposition}

\section{Open problems}

\begin{enumerate}
\item What is the plane-projecting dimension of $K_n$?

\item Find tight bounds for the maximum number of edges in a $n$-vertex graph that is plane-projectable in $\mathbb{R}^3$.

\item Is every graph of maximum degree 6 plane-projectable in three dimensions?

\item Is $pdim(G)=O(\sqrt{\Delta(G)})$?

\item Is it true that $pdim(G)$ is at most the smallest $d$ such that $\dbinom{d}{2} \geq \bar{\theta}(G)$?

\end{enumerate}

Things to do:
\begin{enumerate}
    \item Must the union of 2 caterpillar forest graphs always be planar? [No. Counter example $K_{3,3}$]
    \item Can every graph of max degree 6 be decomposed into a planar graph and one or two caterpillar graphs?
    \item 2-degenerate graphs and 3-degenerate graphs?
    \item Smallest graphs which are not 2-plane projectable/3-plane projectable
    \item Closure under minor or topological minor [Or counter-examples]
    \item $\sqrt{2ca(G)}\leq pdim(G) \leq ca(G)$
    Example of a graph with $pdim(G)=k$ and $ca(G)=O(k^2)$
    \item $K_{n,n}$
    
\end{enumerate}

\bibliographystyle{plain}
\bibliography{bib-thickness}

\end{document}